\documentclass[11pt]{article}
\usepackage[utf8]{inputenc}
\usepackage[american]{babel}
\usepackage[T1]{fontenc}
\usepackage{amssymb}
\usepackage{amsmath}
\usepackage{listings}
\usepackage{amsthm}
\usepackage{graphicx}
\usepackage{color}
\usepackage{xspace}
\usepackage{todonotes}
\usepackage{enumerate}
\usepackage{xargs}
\usepackage{tabularx}
\usepackage{multirow}
\usepackage{listings}
\usepackage{stmaryrd}
\usepackage{mathtools}  % :=, \mathclap and several bugfixes for amsmath
\usepackage{float}

\usepackage{tikz}

\usepackage{framed}
\usepackage{xcolor}
\colorlet{shadecolor}{yellow}
  
\usepackage[plainpages=false]{hyperref}
\hypersetup{
	unicode=false,
	colorlinks=true,
	linkcolor=blue,
	citecolor=blue,
	filecolor=blue,
	urlcolor=blue,
	pdfauthor={Sebastian Abshoff, Andreas Cord-Landwehr, Matthias Fischer, Daniel Jung, Friedhelm Meyer auf der Heide},
	pdfsubject={Gathering a Closed Chain of Robots on a Grid},
	pdftitle={Gathering a Closed Chain of Robots on a Grid},
	pdfkeywords={Gathering problem, Autonomous robots, Distributed algorithms, Local algorithms, Mobile agents, Runtime bound, Swarm formation problems}
}

\newtheorem{theorem}{Theorem}
   \newtheorem{lemma}{Lemma}
\newtheorem{definition}{Definition}

\newcommandx*{\LDAUOmicron}[2][1=@pkling_false]{\mathcal{O}\ifthenelse{\equal{#1}{small}}{\bigl(#2\bigr)}{\left(#2\right)}}
\newcommandx*{\LDAUomicron}[2][1=@pkling_false]{\mathrm{o}\ifthenelse{\equal{#1}{small}}{\bigl(#2\bigr)}{\left(#2\right)}}
\newcommandx*{\LDAUOmega}[2][1=@pkling_false]{\Omega\ifthenelse{\equal{#1}{small}}{\bigl(#2\bigr)}{\left(#2\right)}}
\newcommandx*{\LDAUomega}[2][1=@pkling_false]{\omega\ifthenelse{\equal{#1}{small}}{\bigl(#2\bigr)}{\left(#2\right)}}
\newcommandx*{\LDAUTheta}[2][1=@pkling_false]{\Theta\ifthenelse{\equal{#1}{small}}{\bigl(#2\bigr)}{\left(#2\right)}}

\newcommandx*{\set}[2][2=@pkling_false]{\left\{#1\ifthenelse{\equal{#2}{@pkling_false}}{}{\;\middle|\;#2}\right\}}

\newcommand{\calO}{\mathcal{O}}

\newcounter{algorithmNumber}

\newcounter{actionNumber}

\newcounter{symactionNumber}

\newcommand{\viewingradius}{11}
\newcommand{\pipelininginterval}{13}
\newcommand{\crossdistance}{3}
\newcommand{\figscale}{1}

\makeatletter
\xspaceaddexceptions{\check@icr}
\makeatother

\title{Gathering a Closed Chain of Robots on a Grid}
\author{
    Sebastian Abshoff
    \and
    Andreas Cord-Landwehr
    \and
    Matthias Fischer
    \and
    Daniel Jung
    \and
    Friedhelm Meyer auf der Heide
}

\date{
    Heinz Nixdorf Institute \& Computer Science Department\\[0.2em]
    University of Paderborn (Germany)\\[0.2em]
    Fürstenallee 11, 33102 Paderborn\\[0.2em]
\quad\\[0.2em]
    \texttt{\{abshoff,cola,mafi,daniel.jung,fmadh\}@uni-paderborn.de}
}

\begin{document}
\bibliographystyle{alphadin}
\maketitle
\thispagestyle{empty}
\begin{abstract}
    We consider the following variant of the two dimensional gathering problem for swarms of robots:
    Given a swarm of $n$ indistinguishable, point shaped robots on a two dimensional grid.
    Initially, the robots form a closed chain on the grid and must keep this connectivity during the whole process of their gathering.
    Connectivity means, that neighboring robots of the chain need to be positioned at the same or neighboring points of the grid.
    In our model, gathering means to keep shortening the chain until the robots are located inside a $2\times 2$ subgrid.
    Our model is completely local (no global control, no global coordinates, no compass, no global communication or vision, \ldots).
    Each robot can only see its next constant number of left and right neighbors on the chain.
    This fixed constant is called the \emph{viewing path length}.
    All its operations and detections are restricted to this constant number of robots.
    Other robots, even if located at neighboring or the same grid point cannot be detected.
    Only based on the relative positions of its detectable chain neighbors, a robot can decide to obtain a certain state.
    Based on this state and their local knowledge, the robots do local modifications to the chain by moving to neighboring grid points without breaking the chain.
    These modifications are performed without the knowledge whether they lead to a global progress or not.
    We assume the fully synchronous $\mathcal{FSYNC}$ model.
    For this problem, we present a gathering algorithm which needs linear time.
    This result generalizes the result from \cite{hopper}, where an open chain with specified distinguishable (and fixed) endpoints is considered.
\end{abstract}
\section{Introduction}
    Over the last years, there was a growing interest in problems related to the creation of formations by autonomous robots.
    A benchmark problem is gathering:
    Given a configuration of $n$ autonomous robots in the plane or on the grid, they have to gather in one position not specified beforehand.
    In this paper, we consider a closed chain of $n$ robots on a two-dimensional grid.
    Initially, the robots form an arbitrary closed chain and keep this connectivity during the whole process of gathering.
    Neighboring robots of the chain need to be positioned at horizontally or vertically neighboring points or the same point of the grid.
    We assume that initially no two chain neighbors are located at the same grid point.
    In our model, gathering means to keep shortening the chain until the robots have gathered.
    Our model is completely local (no global control, no global coordinates, no compass, no global communication or vision, \ldots).
    
    The robots have only local vision.
    Each of the robots can only see the subchain, consisting of itself and its next $\viewingradius$ chain neighbors, including their relative positions, in both directions along the chain.
    We call this the robot's \emph{viewing range}, while $\viewingradius$ being the \emph{viewing path length}.
    The locality of our strategy is given by restricting all recognitions and actions of a robot to its constant sized viewing range.
    Other robots cannot be detected by a robot, even if they are located at the same or neighboring grid points and furthermore do not restrain its movements.
    The robots are indistinguishable, i.e., they do not have ids.
    
    We use the fully synchronous time model $\mathcal{FSYNC}$.
    Time is subdivided into equally sized rounds of constant lengths.
    In every round all robots simultaneously execute their operations in the common \emph{look-compute-move} model \cite{Cohen:2004a} which divides one operation into three steps.
    In the \emph{look} step the robot gets a snapshot of the current scenario from its own perspective, restricted to its constant sized viewing range.
    During the \emph{compute} step, the robot computes its action, and eventually performs it in the \emph{move} step.
    
    A robot can change its position to its horizontal, vertical or diagonal neighboring grid points, while our algorithm ensures that this does not disconnect the chain.
    We say that the robot \emph{hops} or performs a \emph{hop} to the according position.
    A \emph{hop} can modify the structure of the chain. Figure~\ref{fig:robot-actions_ex} shows an example for this.
    If after the hop two chain neighbors are located at the same grid point, then from then on, we want them to behave like one robot.
    In our model, their neighborhoods are merged and one of both is removed.
    This \emph{merge} operation shortens the chain by removing robots.
    This is also our progress measurement of the gathering.
    (Note that robots which are located at the same grid point, but not being neighbors on the chain, are not merged/removed by this operation!)
    \begin{figure}[h]
        \centering
        \includegraphics[scale=\figscale]{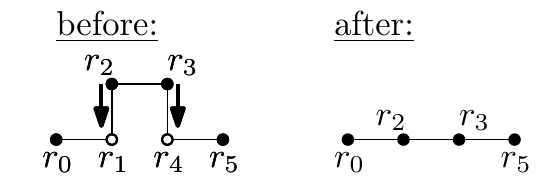}
        \caption{Significant example for shortening the chain (\emph{merge}): At the same time both $r_2$ and $r_3$ hop downwards while the other robots in the figure do not move.
        After the hops, $r_2$ is located at the same position as $r_1$ and $r_3$ is located at the same position as $r_4$.
        Then the neighborhoods merge and the white robots can be removed.}
        \label{fig:robot-actions_ex}
    \end{figure}
    The robots only have a constant memory.
    
    The strict locality of our robot model makes gathering challenging.
    For example, a given global vision or alternatively just the knowledge of a global compass, would make the gathering problem easier, because the robots could compute the center of the globally smallest enclosing square and just move to this point (global vision) or all robots without any local neighbors in front of them could simply move for example to the south-eastern direction and would finally meet (global compass).

    We present a strategy without such global information that solves the gathering problem in time $\calO(n)$.
    We say, gathering is solved when all robots are located within a $2\times 2$ square, since in our time and robot model this symmetry cannot be broken.
    (Note that our algorithm ensures that such symmetries can only occur if the initial number of robots was even.)
    Our result is asymptotically optimal for worst case closed chains.

    As mentioned, we want to do the gathering by shortening the chain.
    For the problem of shortening an \emph{open} chain between two fixed (unmovable) robots, the first shown runtime bound was $\calO(n^2\log(n))$ \cite{gtm}.
    Later, this has been improved \cite{hopper}:
    In the Euclidean plane, the \emph{Hopper} strategy delivers a $\sqrt{2}$-approximation of the shortest open chain in time $\calO(n)$.
    Restricted to a grid, the \emph{Manhattan Hopper} strategy delivers an optimal solution in time $\calO(n)$.

    Assuming that both endpoints of the open chain are located at the same position, this problem looks similar to our closed chain problem.
    The gathering of an open chain would furthermore be simple in general, as the endpoints are always locally distinguishable and would simply sequentially hop onto their inner neighbors.
    In our closed chain problem, none of the robots are distinguishable, so these solutions cannot be applied.

    Nevertheless, we adopt concepts of the Hopper strategies.
    Roughly speaking, these strategies work as follows:
    A fixed endpoint of the open chain sends out a state which then moves along the chain towards the other endpoint of the chain.
    A robot that currently has the state can execute some certain action which then leads to a stepwise shortening of the chain.
    The endpoint sends out such a moving state multiple times, until the robots' actions have, in case of the Manhattan Hopper, optimally shortened the chain and reduced the number of robots to a minimal needed number.

    Since the robots themselves are indistinguishable in our closed chain model, we instead do a distinction depending on the relative positions of the constant number of chain neighbors within their local viewing ranges.
    Using this criterion, the robots elect themselves to obtain the special role of sending out the states.
    Using this approach, we have to deal with several difficulties.
    \begin{itemize}
        \item At the same time, more than just one robot obtain the special role.
        \item Only some of them will actually achieve a shortening of the chain, but they cannot detect this because of the restricted local viewing ranges.
        \item The ones that do not achieve a shortening of the chain shall not hinder the work of the others.
        \item Because robots are moving, at different times different robots obtain the special role. So new states start from different robots.
    \end{itemize}
\section{Related Work.}
    There is a vast literature on robot problems, researching how specific coordination problems can be solved by a swarm of robots given a certain limited set of abilities.
    The robots are usually point-shaped (hence collisions are neglected) and positioned in the Euclidean plane.
    They can be equipped with a memory or are \emph{oblivious}, i.e., the robots do not remember anything from the past and perform their actions only on their current views.
    If robots are anonymous, they do not carry any IDs and cannot be distinguished by their neighbors.
    Another type of constraint is the compass model:
    If all robots have the same coordinate system, some tasks are easier to solve than if all robots' coordinate systems are distorted.
    In \cite{gathering-compasses,Izumi2012} a classification of these two and also of dynamically changing compass models, as well as their effects regarding the gathering problem in the Euclidean plane, is considered.
    The operation of a robot is considered in the \emph{look-compute-move model} \cite{Cohen:2004a}.
    How the steps of several robots are aligned is given by the \emph{time model}, which can range from an asynchronous $\mathcal{ASYNC}$ model (for example, see \cite{Cohen:2004a}), where even the single steps of the robots' steps may be interleaved, to a fully synchronous $\mathcal{FSYNC}$ model (for example, see \cite{localgathering}), where all steps are performed simultaneously.
    A collection of recent algorithmic results concerning distributed solving of basic problems like gathering and pattern formation, using robots with very limited capabilities, can be found in \cite{flocchinioverview}.
    
    One of the most natural problems is to gather a swarm of robots in a single point.
    Usually, the swarm consists of point-shaped, oblivious, and anonymous robots. The problem is widely studied in the Euclidean plane.
    Having point-shaped robots, collisions are understood as merges/fusions of robots and interpreted as gathering progress (cf.\ \cite{MINCH}).
    In \cite{gathering-icalp} the first gathering algorithm for the $\mathcal{ASYNC}$ time model with multiplicity detection (i.e., a robot can detect if other robots are also located at its own position) and global views is provided.
    Gathering in the local setting was studied in \cite{localgathering}.
    In \cite{impossibilityofgathering} situations when no gathering is possible are studied.
    The question of gathering on graphs instead of gathering in the plane was considered in \cite{practicalrendevouzaktuell, rendezvousingraphen, gatheringOnRing}.
    In \cite{Stefano2013} the authors assume global vision, the $\mathcal{ASYNC}$ time model and furthermore allow unbounded (finite) movements.
    They show optimal bounds concerning the number of robot movements for special graph topologies like trees and rings.
    
    Concerning the gathering on grids,
    in \cite{gatheringongrids} it is shown that multiplicity detection is not needed and the authors further provide a characterization of solvable gathering configurations on finite grids.
    In \cite{OptExactGatheringGrids2014}, these results are extended to infinite grids, assuming global vision.
    The authors characterize \emph{gatherable} grid configurations concerning exact gathering in a single point.
    Under their robot model and the $\mathcal{ASYNC}$ time model, the authors present an algorithm which gathers \emph{gatherable} configurations
    optimally concerning the total number of movements.
         
    In the $\mathcal{FSYNC}$ time model, the total running time is a quality measurement of an algorithm.
    In \cite{gatheringthetanquadrat}, an $\calO(n^2)$ runtime bound is shown for an algorithm working on point shaped robots without compass in the Euclidean plane in the $\mathcal{FSYNC}$ model.
    The vision of the robots is restricted to a local viewing radius of constant length only, instead of a global vision of the whole scenery.
    Because of the restricted vision, the robots cannot compute a global gathering point.
    Instead, the robots synchronously compute the smallest enclosing circle of the robots within their restricted viewing range and then move towards its center.
    The authors also prove that for their algorithm this bound is tight.
    For the problem itself, under this local model, a tight bound for the running time is still unknown.

    Our result of the underlying paper also works in a local robot model and in $\mathcal{FSYNC}$.
    The proven total running time $\calO(n)$ is tight concerning the gathering in our model.

    A preliminary version of this paper was published at \cite{Jung2015a}.
    Compared to our new paper, it had several drawbacks that our new algorithm solves:
    (1) robots needed a global clock,
    (2) robots needed a storage for saving a time stamp value, and
    (3) the target configuration was of arbitrary fixed size and not a $2\times 2$ square.
\section{Basic idea of the Algorithm}\label{sec:basicidea}
    Our measurement of progress is the shortening of the chain.
    Then, after $n-1$ shortenings the gathering is done.
    In our model, the chain can be shortened if two chain neighbors are located at the same grid point.
    Then, one of them is removed and for the other one the neighborhoods of both are combined and assigned to it,
    which results in the shortening of the chain.

    Because a robot's viewing range is restricted to just a constant size, these shortenings cannot be performed easily in general.
    Our algorithm performs two basic operations.
    More precisely, we have to deal with two cases:
    \begin{enumerate}
        \item Some neighboring robots perform a single hop such that afterwards at least two chain neighbors are located at the same position, while preserving the chain connectivity. Afterwards, the shortening as explained above can be performed.
        This is a so called \emph{merge} and further discussed in Subsection~\ref{ssec:merges}.\label{enum:easiershortening}
        \item If on some subchains merges are impossible, we perform so called reshapements to reshape the chain in order to prepare merges in succeeding steps.
        They are explained in Subsection~\ref{ssec:runs}.
    \end{enumerate}
    In Subsection~\ref{ssec:merges}, we now start with case \ref{enum:easiershortening}.
    For the sake of simpler descriptions, we use terms like \emph{horizontal, vertical, downwards, left,\ldots}
    Since our robots do not have a common sense for this, the descriptions/figures are also to be understood in a mirrored or rotated manner.
    \subsection{Merges}\label{ssec:merges}
        As introduced, the chain can be shortened if two chain neighbors are located at the same grid point.
        Figure~\ref{fig:ALG_merge2} shows how our algorithm manages that two chain neighbors get located at the same grid point.
        In the figure, the black robots hop downwards simultaneously.
        Afterwards, assuming that the length $k$ of the black subchain is bigger than $1$, the both outermost of them are located at the same grid point as the bordering white robots.
        Because at both ends the robots are chain neighbors, we can remove the white ones without breaking the chain.
        If $k=1$, then after the hop also both white robots can be removed without breaking the connectivity of the chain.
        We call these operations \emph{merge operations} or shorter: \emph{merges}.
        Merge operations can only be performed if all participating robots can see all black and white robots.
        This is necessary, because else not all of the black robots know that they must perform a hop in order to prevent the chain from breaking.
        So the length $k$ especially cannot be larger than a robot's constant viewing path length.
        \begin{figure}[h]
        \centering
            \includegraphics[scale=\figscale]{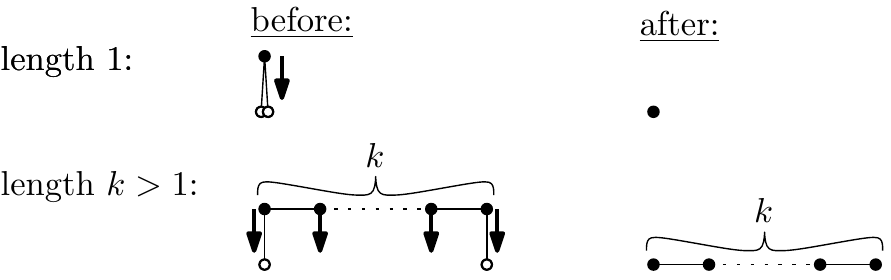}
            \caption{Subchains that allow progress hops (\emph{merge operations}). $k=1$: Here, the white robots actually are located at the same grid point and are drawn with a small distance just for presentation reasons. The value of $k$ indicates the length of the black subchain. $k$ is upper bounded by a robot's constant viewing path length.}
            \label{fig:ALG_merge2}
        \end{figure}
         
        It may happen that at the same time more than just one merge can be performed on different parts of the chain.
        Then, two cases need a closer look (We refer to the black and white robots of Figure~\ref{fig:ALG_merge2}.):
        \begin{enumerate}
            \item The merge subchains, including the white robots, overlap by two robots.\label{enum:oneoverlap}
            \item The merge subchains, including the white robots, overlap by three robots.\label{enum:twooverlap}
        \end{enumerate}
        An example for \ref{enum:oneoverlap}) can be seen in Figure~\ref{fig:symmerges}.$a)$.
        Here, all black robots perform the same hop as the black ones in Figure~\ref{fig:ALG_merge2}.
        The difference is that afterwards for example the robot $a$ is not located at the same position as $b$ and vice versa.
        So, there the chain cannot be shortened.
        But the outermost white robots of Figure~\ref{fig:symmerges}.$a)$ do not move.
        Here, the shortening can be performed.

        An example for \ref{enum:twooverlap}) is shown in Figure~\ref{fig:symmerges}.$b)$.
        Here, the robot $r$ belongs to the black robots of subchain $1$ and $2$.
        This requires an additional rule, because concerning subchain $1$ $r$ would hop downwards, but concerning subchain $2$ it would hop
        to the left.
        Instead, we let $r$ perform a diagonal hop to the lower left, while the other black robots perform their usual hops.
        Afterwards, $r,a,b$ are located at the same position and $a,b$ are removed without breaking the chain.

        Because in all three cases the merge operations lead to the removal of robots, i.e.., shortening of the chain, all satisfy our measurement of progress.
        \begin{figure}[h]
        \centering
            \includegraphics[scale=\figscale]{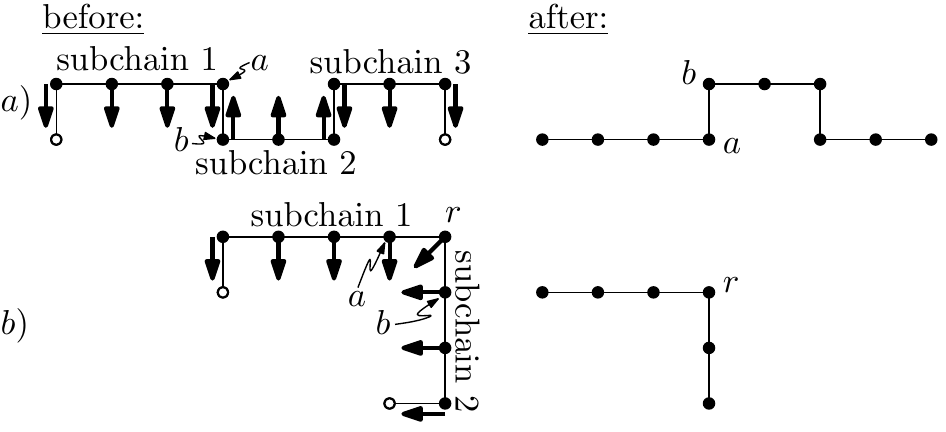}
            \caption{Cases for merges of subchains which are not node disjoint. (significant examples)}\label{fig:symmerges}
        \end{figure}
    \subsection{Reshapement of the chain, done by runners}\label{ssec:basicruns}
        \begin{figure}[h]
        \centering
            \includegraphics[scale=\figscale]{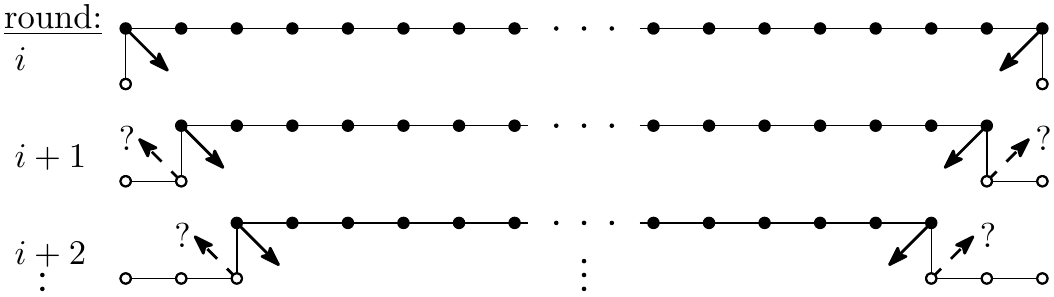}
            \caption{If the length of the black subchain (i.e., the value of $k$ in Figure~\ref{fig:ALG_merge2}) is larger than the robots' viewing path length, we shorten it by letting the outermost black robots perform diagonal hops. We indicate such hops by diagonal arrows.}
            \label{fig:goodpair_ex_simple_norunners}
        \end{figure}
        If because of the robot's restricted viewing ranges nowhere on the whole chain a \emph{merge} can be performed, then we call the chain a \emph{Mergeless Chain}.
        We perform so called \emph{reshapements} to reshape the chain in order to prepare merges in succeeding steps:
        We continuously let the outermost robots of the black subchain perform diagonal hops in order to make it shorter (cf.\ Figure~\ref{fig:goodpair_ex_simple_norunners}) until it becomes short enough for enabling the operation of Figure~\ref{fig:ALG_merge2}.
        We say, these hops \emph{reshape} the subchain.
        This task introduces two major challenges:
        \begin{enumerate}
            \item When the outermost black robots of Figure~\ref{fig:goodpair_ex_simple_norunners} decide to start the reshapement (round $i$), within their restricted viewing range they do not know if their task of performing diagonal hops will lead to merges.
            \item In the following rounds, the subchain which an outermost black robot can see, in general locally looks the same as the one, visible by its white neighbor.
            So, we have to ensure that the reshapement is continued by the outermost black robots instead of by their white neighbors.
        \end{enumerate}
        In the following, we tackle these challenges.
        For this, we first introduce a certain state of a robot, which we call the \emph{run} state.
        We call a robot with active run state a \emph{runner}.
        Robots can achieve this state in two different ways:
        \begin{enumerate}
            \item (\emph{start runstate}) If the local subchain within a robot's viewing range has a certain shape, then the robot decides on its own to generate the run state.
                We say, such a robot starts a run.
                Based on the shape of the local subchain, the run state gets a fixed moving direction along the chain.
                A robot can start and store up to two run states at the same time.
                Figure~\ref{fig:runstartingrobots} shows how the local shapes must look like.\label{enum:runstart}
            \item (\emph{move runstate}) A runner $R(S)$ can move the run state $S$ to its chain neighbor $r'$ in moving direction of $S$.
            We say, the run state has moved from $R(S)$ to $r'$, while its in \ref{enum:runstart})\ initially set moving direction always remains unchanged. Afterwards, $r'$ is identified by $R(S)$.\label{enum:runmovement}
        \end{enumerate}
        Once a run state has been started in \ref{enum:runstart}), \ref{enum:runmovement})\ is executed in every of the following rounds.
        This means that the run moves along the chain at constant speed and in the initially settled moving direction.
        \begin{figure}[h]
            \centering
            \includegraphics[scale=\figscale]{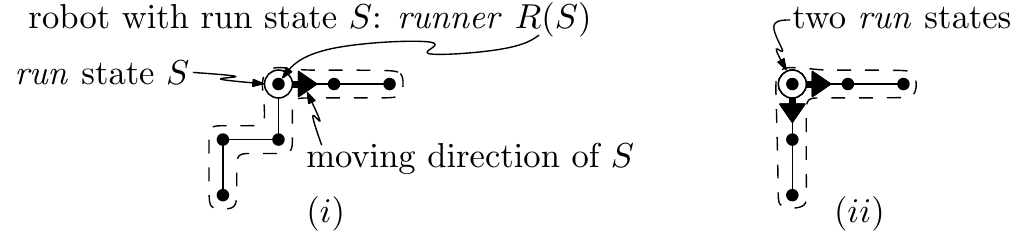}
            \caption{The encircled robots decide to start the run states, only based on the shape of the marked subchain within their viewing range.
            The large arrow heads indicate the moving direction of the runs.
            This is the notation, we will use for marking a runner. $R(S)$ identifies the runner/robot which currently has the run state $S$;
            In $(ii)$, the encircled robot is the endpoint of a horizontally and a vertically aligned subchain at the same time. Here, we must start two runs, moving in both directions along the chain.}
            \label{fig:runstartingrobots}
        \end{figure}
        \begin{figure}[h]
        \centering
            \includegraphics[scale=\figscale]{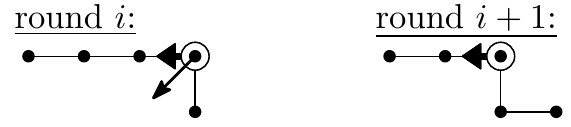}
            \caption{If in round $i$ the local subchain looks like this, then the runner performs a diagonal hop and the run state moves to the next robot in moving direction.}
            \label{fig:ALG_hop_basic}
        \end{figure}
        A runner can perform a diagonal reshapement hop (and afterwards move the run state to its chain neighbor in moving direction) if the chain locally looks like in Figure~\ref{fig:ALG_hop_basic} (round $i$).
        Figure~\ref{fig:runmeeting_simple}.$a)$ shows sequent operations.
        This solves the problem, we have had in Figure~\ref{fig:goodpair_ex_simple_norunners}, where it was impossible to locally decide by the robots, which of them has to perform the reshapement hop.
        \begin{figure}[h]
        \centering
            \includegraphics[scale=\figscale]{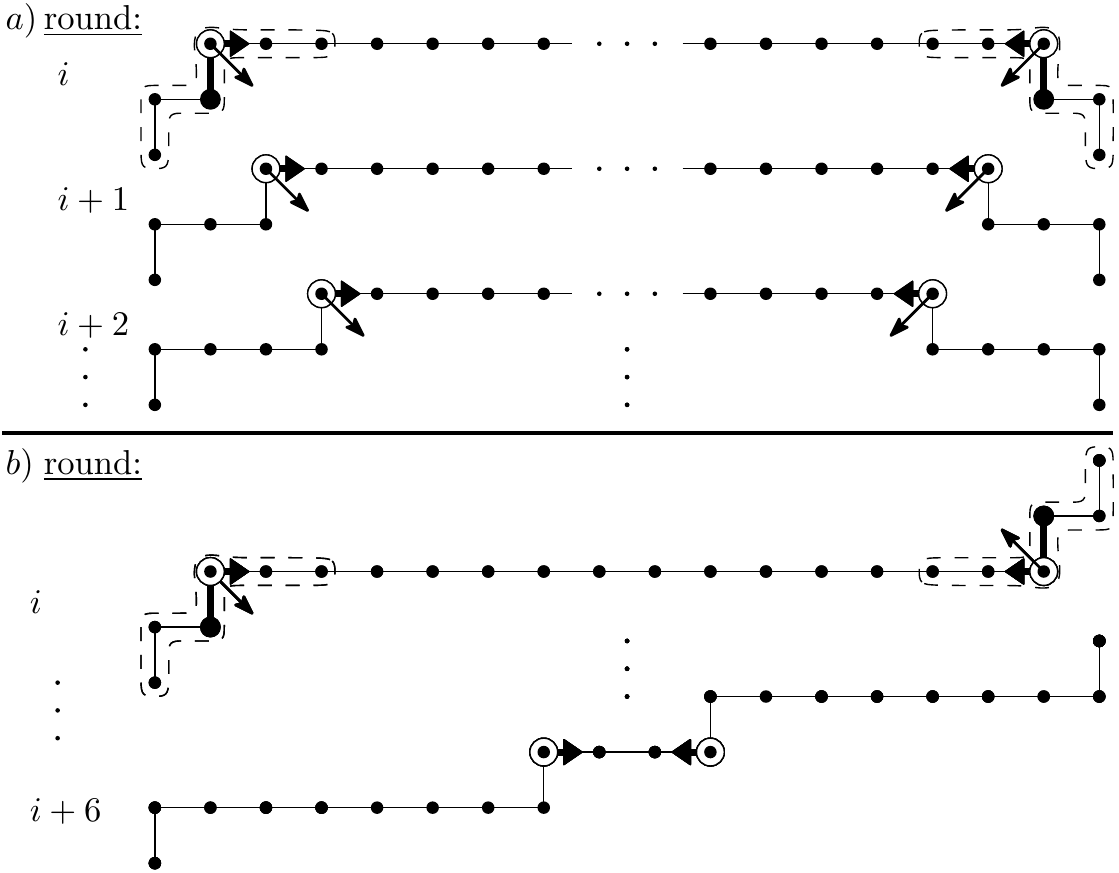}
            \caption{Run pairs, started at both ends of horizontally aligned subchains.
            $a,b)$: In round $i$, the new runs start (cf.\ Figure~\ref{fig:runstartingrobots}).
            Afterwards, the actions of Figure~\ref{fig:ALG_hop_basic} are repeatingly executed.
            The runs are moving closer and closer together.
            $a)$: The run pair is a \emph{good pair}. If the runs have moved close enough, then a merge can be performed.
            $b)$: No merge can be performed. Then the runs just pass along each other. (Details: Figure~\ref{fig:runpass})}
            \label{fig:runmeeting_simple}
        \end{figure}
        \begin{figure}[h]
        \centering
            \includegraphics[scale=\figscale]{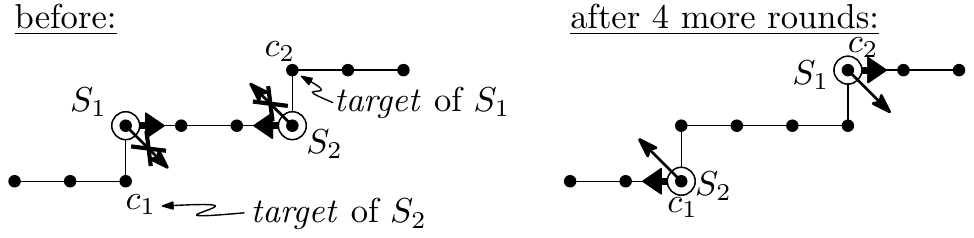}
            \caption{\emph{Run passing} operation: The runs $S_1,S_2$ do not enable a merge.
            If their distance along the chain is less or equal $\crossdistance$, then they pass each other by only keeping moving but without making the runners perform diagonal hops.
            Afterwards, i.e., when $S_1,S_2$ have reached their target robots/corners $c_2,c_1$, they return to normal operation.
            }
            \label{fig:runpass}
        \end{figure}
        But not all run pairs, started at the endpoints of a horizontally aligned subchain, do actually enable a merge.
        We distinguish \emph{good pairs} and \emph{non good pairs}.

        \underline{\emph{good pairs}:}
        (Cf.\ Figure~\ref{fig:runmeeting_simple}.$a)$) If the exterior neighbors (the fat robots in the figure) of the newly started run pair are both located on the same side of the subchain, then the run pair is called a \emph{good pair}.
        These pairs enable a merge if they have been moving close enough together and then terminate.

        \underline{\emph{non good pairs}:}
        Figure~\ref{fig:runmeeting_simple}.$b)$ shows the opposite case, i.e., the fat robots are located at different sides of the subchain.
        These pairs do not enable a merge.
        We let the runs of such pairs pass along each other.
        Figure~\ref{fig:runpass} shows how this is performed:
        At the time when their distance (i.e., the number of edges on the subchain connecting both) is $\crossdistance$ or less, they only keep moving along the chain, but the runners do not perform reshapement hops.
        This is repeated until $S_1$ is located at its target robot $c_2$ (then also $S_2$ is located at its target $c_1$).
        We call this the \emph{run passing} operation.
        Afterwards, the normal reshapement operations are continued.
        Note that depending on whether the distance between $S_1$ and $S_2$ was odd or even, it happens that during this passing process
        at one time both runs are located at the same robot.
        Then, this robot handles both runs separately according to their movement directions.
    \subsection{Parallelizing runs: \emph{Pipelining}}\label{ssec:basicpipelining}
        The total process of gathering needs the work of many good pairs.
        For sake of a short running time, we let some of them work in parallel.
        Every constant number of $L=\pipelininginterval$ rounds, all robots simultaneously check if they can start new runs (cf.\ Figure~\ref{fig:runstartingrobots}) and if so, they do so.
        We will show that this procedure ensures that even if multiple good pairs are nested into each other, different good pairs will enable different merges.
        This is what we call \emph{pipelining}.
        \begin{figure}[h]
        \centering
            \includegraphics[scale=\figscale]{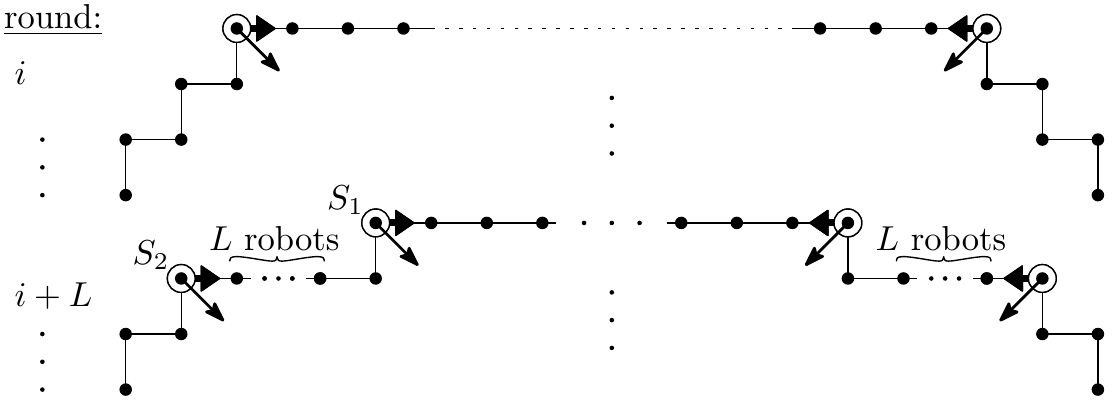}
            \caption{Pipelining of runs. New runs are started every $L=\pipelininginterval$ rounds.}
            \label{fig:pipelining_ex}
        \end{figure}
        Figure~\ref{fig:pipelining_ex} shows an example for this.
        Because $S_1, S_2$ are moving in the same direction, we call them \emph{sequent} runs, while relative to their moving direction
        $S_1$ is located \emph{in front of} $S_2$.
        The \emph{distance} between them equals the number of edges on the subchain connecting both.

        As runs are moving with constant speed, the runs of the inner good pair will meet and as the result enable a merge and stop, first.
        Then, obviously the outer good pair will also enable a (different) merge, some rounds later.
    \subsection{Stopping runs:}\label{ssec:basicstopping}
        In order to ensure that the pipelining works correctly, we let a runner $R(S)$ stop/terminate its run $S$, if one of the following conditions is true:
        \begin{enumerate}
            \item It can see the next sequent run in front of it (This happens if sequent runs have come too close to each other, e.g., because of merge operations.).\label{enum:basictooclosestop}
            \item It can see the end of the horizontally aligned subchain it is located on, in front of it.
            \item It was part of a merge operation.
            \item While it performs the run passing operation of Figure~\ref{fig:runpass}, the target corner is removed (This can happen because of a merge operation.).\label{enum:basicrunpassstop}
        \end{enumerate}

        We give some more detailed explanations concerning some of the above conditions:
        \begin{itemize}
            \item[\ref{enum:basictooclosestop})]
                For example, because of merges two sequent runs may come too close to each other, which then might hinder the pipelining.
                The affected runners can detect this on their own.
                The criterion for this is that the next sequent run in front of them becomes visible.
                Then, the termination condition \ref{enum:basictooclosestop}) matches and the run behind stops.
            \item[\ref{enum:basicrunpassstop})]
                We assume, that in Figure~\ref{fig:runpass} to the right of robot $c_2$ another run $S_3$, moving in the same direction as $S_2$, is located.
                During the rounds in which $S_1$ and $S_2$ are performing their run passing operation, $S_3$ keeps moving towards $c_2$.
                Now, it may happen that because of the reshapements of $S_3$, $c_2$ becomes part of a merge operation.
                Then, $c_2$ would hop downwards such that the corner shape does not exist anymore.
                Because this corner has been the target of $S_1$, $S_1$ could not continue its reshapements after the run passing, so it terminates.
        \end{itemize}
    \subsection{Correctness and running time}
        If every round a merge can be performed, then the time needed for the gathering is obviously upper bounded by $n$,
        with $n$ being the number of robots.
        If no merge can be performed, then the shape of the chain is reshaped by runs.
        We have shown that good pairs enable merges.
        So we need to show that if nowhere on the whole chain a merge is possible, always at least one new good pair can be started.

        A single good pair needs at most $n$ rounds until a merge can be performed.
        In Subsection~\ref{ssec:basicpipelining}, we have already noticed that if new good pairs are started every $L=\pipelininginterval$ rounds,
        then different good pairs lead to different merges.
        The last good pair is started at round $L\cdot n$ and finishes its work after at most $n$ further rounds.
        As $L$ is a constant, we then get to the total linear running time $\in\calO(n)$.
\section{Algorithm in Detail}\label{sec:algdetail}
    Now, we explain the complete strategy.
    In Section~\ref{sec:basicidea}, we let runs move only along strictly horizontally aligned subchains.
    But such subchains do not always exist.
    We need to generalize them to so called \emph{quasi lines} (see Definition~\ref{def:quasiline}).
    The main difference to the description of Section~\ref{sec:basicidea} is that we now let runs move along such \emph{quasi lines}.
    Figure~\ref{fig:quasiline_ex2} gives an example of a quasi line.
    \begin{figure}[h]
        \centering
        \includegraphics[scale=\figscale]{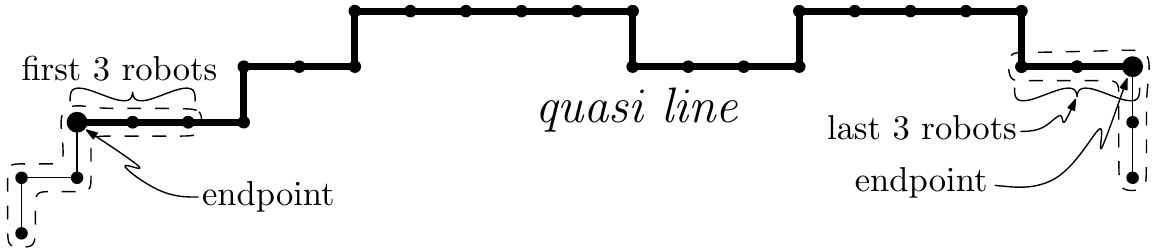}
        \caption{Example of a quasi line. The fat robots are its endpoints.}
        \label{fig:quasiline_ex2}
    \end{figure}
    \begin{definition}[quasi line]\label{def:quasiline}
        We call a subchain a horizontal \emph{quasi line}, if the following points hold:
        \begin{enumerate}
            \item At least its first and last three robots are horizontally aligned.
            \item All its subchains of horizontally aligned robots contain at least three robots.
            \item All its subchains of vertically aligned robots contain at most two robots.
        \end{enumerate}
        In a Mergeless Chain, at both ends a subchain of Figure~\ref{fig:runstartingrobots} in a matching rotation or reflection occurs. (If the chain is not mergeless, then the subchains outside the quasi line's endpoints may also have other shapes than these.)

        The definition of a vertical \emph{quasi line} follows analogously.
    \end{definition}
%    At every end of a quasi line, one of the run starting subchains of Figure~\ref{fig:goodpair_ex_simple}.$a)$ occurs.
    In Figure~\ref{fig:quasiline_ex2}, the fat subchain, connecting and including the fat robots at its endpoints, is called a quasi line.

    Having introduced quasi lines, we now have to analyze, how this affects \emph{merges}, \emph{runs} and \emph{pipelining} we have had explained in Section~\ref{sec:basicidea}.

    Merges remain exactly the same as in the basic description (Subsection~\ref{ssec:merges}), i.e., they
    are only performed with black subchains, solely consisting of strictly horizontally aligned robots as shown in Figure~\ref{fig:ALG_merge2}.
    So, we can continue with \emph{Mergeless Chains} and the analog to Subsection~\ref{ssec:basicruns}.
    \subsection{Reshapement of the chain, done by runners}\label{ssec:runs}
        The main approach remains the same as in Subsection~\ref{ssec:basicruns}.
        I.e., if in Figure~\ref{fig:ALG_merge2} the black subchain is longer than the robots viewing path length, we use reshapement hops of runners for shortening this black subchain until a merge becomes possible.
        We start new runs at the same subchains as in the basic description (cf.~Figure~\ref{fig:runstartingrobots}).
        The only difference is that now these subchains are connected by a quasi line.
        Because of this, as one can see in the example of Figure~\ref{fig:goodpair_ex2_quasiline}, the runs now have to move several steps
        along the chain until arriving at the endpoints of the subchain, bordered in the figure, which needs to be shortened for performing the merge.
        For this movement, depending on the local shape of the subchain, we require additional run operations:
        Figure~\ref{fig:ALG_hop_simple2}.$a)$ is the basic operation, while $b)$ and $c)$ are the new ones.
        \begin{figure}[h]
        \centering
            \includegraphics[scale=\figscale]{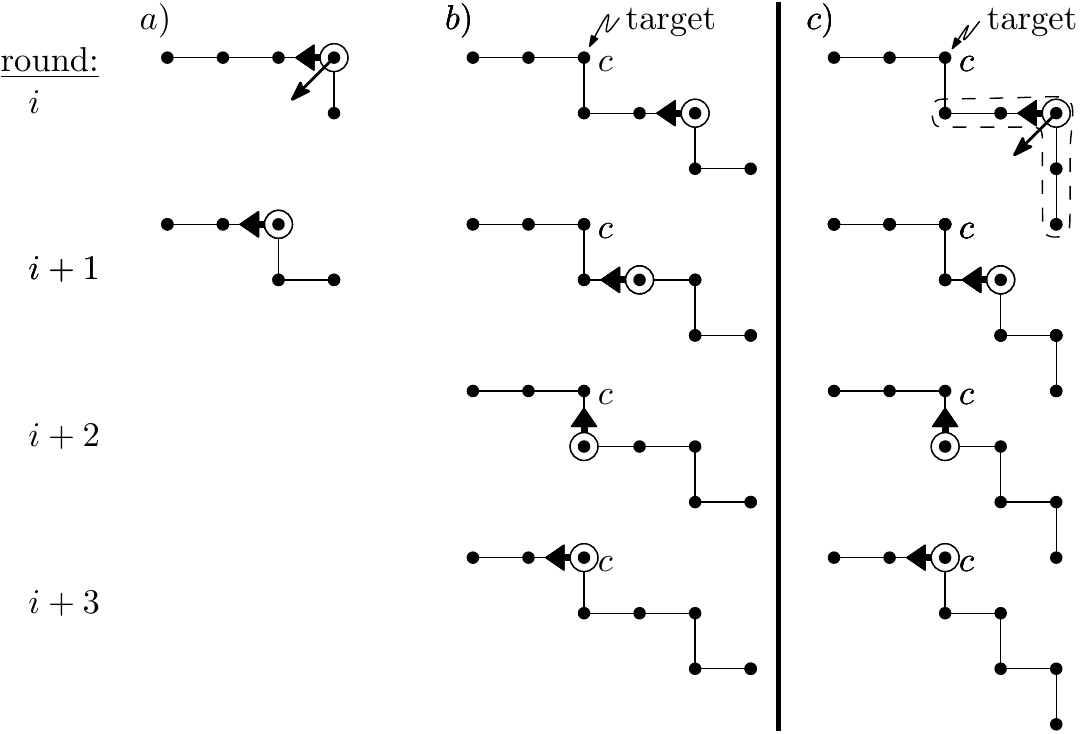}
            \caption{$a)$: Reshapement by a runner. The operation takes only one round. $b)$: No diagonal hops are performed until the target corner $c$ is reached. $c)$: Special case when new runs start: First, perform one diagonal hop, then no diagonal hops until the target corner $c$ is reached.}
            \label{fig:ALG_hop_simple2}
        \end{figure}
        \begin{itemize}
            \item[$a)$] The runner and at least the next 3 robots are located on a straight line. Here, the runner first performs a diagonal hop, then moves the run to the next robot.
            \item[$b)$] The runner and only the next 2 robots are located on a straight line. Then, for 3 times the runners just move the run to the next robot without any diagonal hops. Afterwards, it is located at the target corner $c$.
            \item[$c)$] This one is needed at most once for a new run, if started at the subchain of Figure~\ref{fig:runstartingrobots}.$(ii)$.
        \end{itemize}
        New runs always start pairwise at both endpoints of a quasi line.
        \emph{Good pairs} of runs are defined analogously to Subsection~\ref{ssec:basicruns} by the relative position, concerning the quasi line, of the outer chain neighbors of the endpoints of the quasi line (cf.\ Figure~\ref{fig:goodpair_ex1}).
        \begin{figure}[h]
        \centering
            \includegraphics[scale=\figscale]{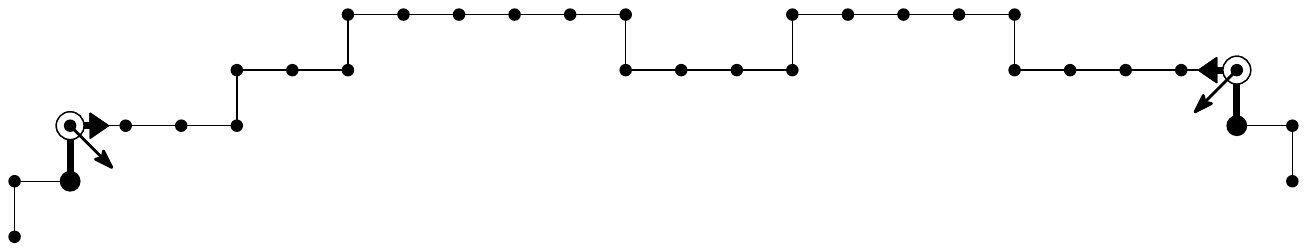}
            \caption{Good pairs of runs, connected by a quasi line. The runs are a good pair if the outer chain neighbors (the fat robots) are both located either downwards or both upwards, i.e., on the same side of the quasi line.}
            \label{fig:goodpair_ex1}
        \end{figure}
        We will show that also with quasi lines, good pairs always enable a merge.
        And as before in Subsection~\ref{ssec:basicruns}, not all run pairs are good pairs.
        \begin{figure}[h]
            \centering
            \includegraphics[scale=\figscale]{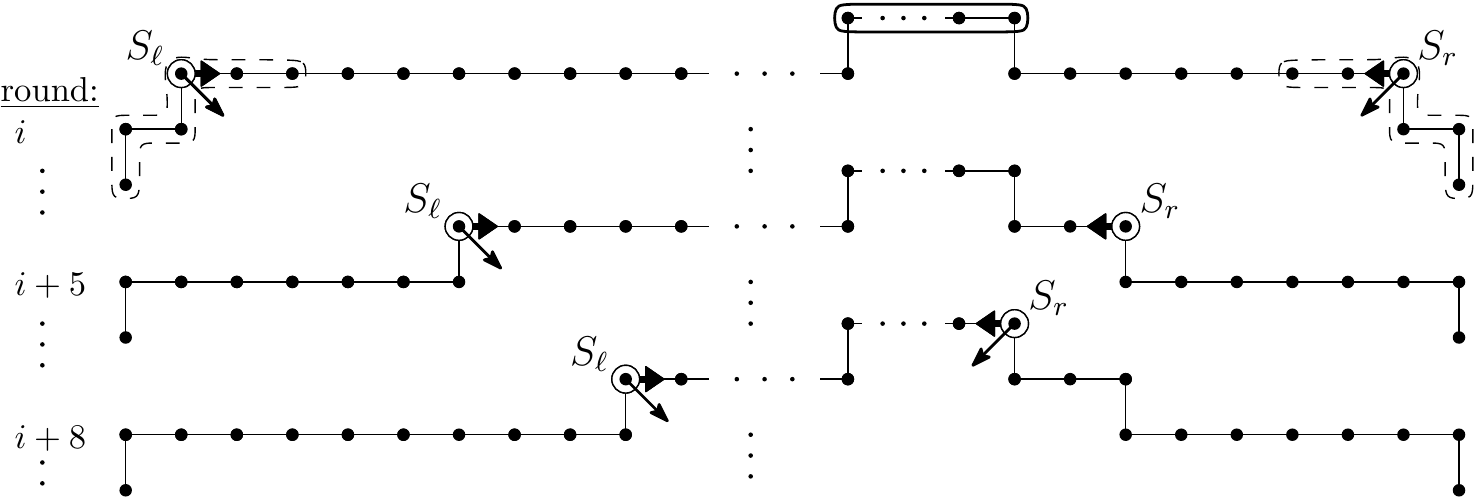}
            \caption{A good pair, started on a quasi line in order to shorten the bordered subchain. Several and different run operations are needed.}
            \label{fig:goodpair_ex2_quasiline}
        \end{figure}
        Figure~\ref{fig:goodpair_ex2_quasiline} shows an example for a good pair on a quasi line:
        The bordered subchain will be reshaped for performing a merge.
        Until round $i+4$, the runners execute only the basic operation of Figure~\ref{fig:ALG_hop_simple2}.$a)$.
        Afterwards, the runner $R(S_r)$ starts operation $b)$, while $R(S_\ell)$ still executes $a)$.
        After operation $b)$ has been processed completely, $S_r$ finally is located at the right end of the bordered subchain (round $i+8$).
        Now, its runners start shortening this subchain by executing $a)$ until the merge can be performed.
        Afterwards, $S_\ell$ is still active and keeps moving and will stop at the latest when an endpoint of the quasi line
        becomes visible.

        If two runs that do not enable a merge meet each other, we execute the run passing operation in a more generalized variant.
        As some of the operations of Figure~\ref{fig:ALG_hop_simple2} now take more than just a single round,
        it may happen that such an operation becomes interrupted by the run passing.
        Then, the target corners for the run passing are settled with respect to the situation when the interrupted operation
        has been started.
        Figure~\ref{fig:crossconst} shows an example:
        In round $i$ the runner $R(S_1)$ starts the execution of operation $b)$ of Figure~\ref{fig:ALG_hop_simple2}.
        In round $i+2$, this operation is still not finished, but the distance between $S_1$ and $S_2$ is $\crossdistance$ such that their runners both have to start the run passing operation.
        Then, the target corner of $S_2$ is the corner $c_1$.
        The target if $S_1$ as before is $c_2$.
        Similar to the basic run passing, both runs do not perform any reshapement hops before arriving at their target corners.
        Other cases, e.g., if the runners of both runs execute operation $b)$, are analog to this one.
        \begin{figure}[h]
            \centering
            \includegraphics[scale=\figscale]{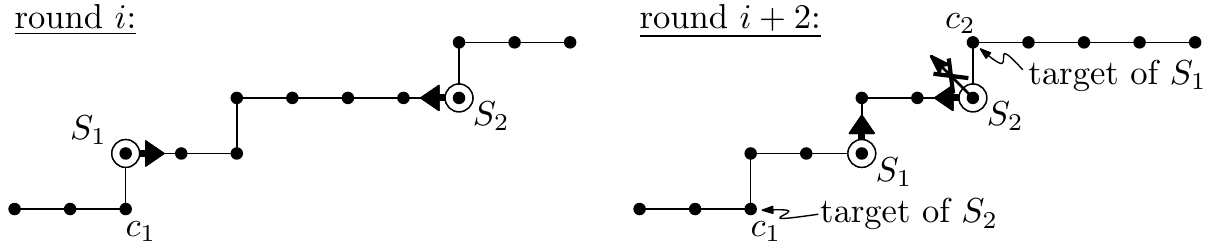}
            \caption{Runs, passing along each other while the runner $R(S_1)$ is executing the operation of Figure~\ref{fig:ALG_hop_simple2}.$b)$.}
            \label{fig:crossconst}
        \end{figure}
    \subsection{Parallelizing runs: \emph{Pipelining}}\label{ssec:pipelining}
        The pipelining works the same as in the basic explanation in Subsection~\ref{ssec:basicpipelining}.
    \subsection{Stopping runs:}
        The basic stop/termination conditions for runs (Subsection~\ref{ssec:basicstopping}) have to be extended for working on quasi lines.
        We let a runner stop/terminate its run, if one of the conditions of Table~\ref{table:runterminate} is true.
        \begin{table}
            \fbox{\parbox{0.97\columnwidth}{
                A runner stops/terminates its run, if at least one of the following conditions is true:
                \begin{enumerate}
                    \item It can see the next sequent run in front of it (This happens if sequent runs have come too close to each other, e.g., because of merge operations.).\label{enum:stoptooclose}
                    \item It can see the endpoint of the quasi line in front of it.\label{enum:stopendpoint}
                    \item It was part of a merge operation.\label{enum:stopmerge}
                    \item While it performs the run passing operation of Figure~\ref{fig:runpass}, the target corner is removed (This can happen because of a merge operation.).\label{enum:stoppass}
                    \item While it performs the operation of Figure~\ref{fig:ALG_hop_simple2}.$b)$ or $c)$, the target corner is removed (This can happen because of a merge operation.).\label{enum:stoplongop}
                \end{enumerate}
            }}
            \caption{Conditions which let a run terminate.}
            \label{table:runterminate}
        \end{table}
        Summarizing, all robots synchronously execute the algorithm, shown in Figure~\ref{fig:algo}.
        \begin{figure}[h]
            \fbox{\parbox{0.97\columnwidth}{
                Every robot $r$ every round checks the following three steps:
                \begin{enumerate}
                    \item \underline{Merge:} If $r$ detects a possible merge within its viewing range then
                        \begin{itemize}
                            \item if $r$ is one of the black robots in Figure~\ref{fig:ALG_merge2}, it hops downwards.
                            \item if afterwards $r$ is located at the same position as one of the white robots, the white one is removed without breaking the chain (If it is located at the same position as both white robots, then both white ones are removed.).
                        \end{itemize}
                    \item \underline{Run Operations:} If $r$ is a \emph{runner}, then
                        \begin{enumerate}
                            \item Its run terminates/stops if any of the conditions of Table~\ref{table:runterminate} is true.
                            \item Runner's Movement and Reshapement
                                \begin{itemize}
                                    \item Run passing:
                                        \begin{itemize}
                                            \item If $r$ is currently in progress of executing the run passing operation (Figure~\ref{fig:runpass} respectively \ref{fig:crossconst}), then this operation is continued.
                                            \item Else, if $r$ can see a run in front of it, such that both are moving towards each other and the distance between them is less or equal $\crossdistance$, then $r$ starts the run passing operation.
                                        \end{itemize}
                                    \item If $r$ is not in progress of passing, then
                                        \begin{itemize}
                                            \item If $r$ is in progress of executing a run operation of Figure~\ref{fig:ALG_hop_simple2}.$b,c)$, which takes more than one round, then this one is continued.
                                            \item Else: $r$ executes the matching new run operation of Figure~\ref{fig:ALG_hop_simple2}.
                                        \end{itemize}
                                \end{itemize}
                        \end{enumerate}
                    \item \underline{Start new runs:} Every $(L=\pipelininginterval)$th round, $r$ checks if it can start a new run:\\
                        If $r$ is one of the encircled robots of Figure~\ref{fig:runstartingrobots}.$(i)$ respectively $(ii)$, then it starts one resp.\ two runs.
                \end{enumerate}
            }}
            \caption{The algorithm.}
            \label{fig:algo}
       \end{figure}
    \section{Correctness and running time}
        In this section, our goal is the proof of the correctness and the linear unning time (Theorem~\ref{thm:runningtime}).
        In our main approach, we want to enable merges by good pairs if else no merge could be performed.
        The following two lemmas prove that this actually works.
        They are the base for the proof of the theorem.
        The proofs of the lemmas can be found in Subsection~\ref{ssec:progresspair_exists} and \ref{ssec:progresspair_merge}.

        Because of the local vision, new runs and maybe also good pairs are started every $L=\pipelininginterval$ rounds, regardless of whether or not at some other location on the chain merges can be performed.
        In the following  analysis, we will only argue with new good pairs which are startet if during the last $L-1$ and the current round on the whole chain no merge has been performed.
        We distinguish such good pairs from others by calling them \emph{progress pairs}.

        \begin{lemma}\label{lem:progresspair_exists}
            Every $L=\pipelininginterval$ rounds either a merge has been performed or else a new \emph{progress pair} is started.
        \end{lemma}
        \begin{lemma}\label{lem:progresspair_merge}
            For \emph{progress pairs} the following properties hold.
            \begin{enumerate}[a)]
                \item Every progress pair enables a merge (after at most $n$ rounds).\label{enum:progpairmerge}
                \item Different progress pairs enable different merges.\label{enum:progpairuniqueness}
            \end{enumerate}
        \end{lemma}
        Using these two lemmas, we can now prove the total linear running time.
        \begin{theorem}\label{thm:runningtime}
            Given a closed chain of $n$ robots. Then, after $\calO(n)$ rounds gathering is done. This is asymptotically optimal.
        \end{theorem}
        \begin{proof}
            We subdivide time into intervals of lengths $L$, where $L$ denotes a number of rounds.
            Merges can be performed during at most $n$ such intervals, because every merge removes at least one robot.
            In all other intervals a new progress pair starts (Lemma~\ref{lem:progresspair_exists}).
            Each of these progress pairs leads to merge (Lemma~\ref{lem:progresspair_merge}.$\ref{enum:progpairmerge})$).
            Because no two of them lead to the same merge (Lemma~\ref{lem:progresspair_merge}.$\ref{enum:progpairuniqueness})$),
            the number of intervals without merges is also upper bounded by $n$.

            By Lemma~\ref{lem:progresspair_merge}, a progress pair needs at most $n$ rounds until it has led to a merge.
            We assume the worst case that in the last of the $2\cdot n$ intervals the last progress pair was started.
            Then the total running time is upper bounded by $2n\cdot L\; +n$, which proves the upper bound of the theorem, because $L$ is a constant.

            In our model, the diameter of the initial configuration provides the worst case lower bound $\Omega(n)$ for any gathering strategy.
        \end{proof} 
        %
    %\subsection{Correctness}\label{ssec:correctness}
    \subsection{Proof of Lemma~\ref{lem:progresspair_exists}}\label{ssec:progresspair_exists}
        \begin{proof}%(Lemma~\ref{lem:progresspair_exists})
            The lemma assumes that the chain is a \emph{Mergeless Chain}.
            Although our viewing path length also allows larger values, for the proof we assume that merges are only possible up to the length $2$ (cf.\ Figure~\ref{fig:ALG_merge2}).
            This suffices, because if a chain is a Mergeless Chain for a bigger length, it also is a Mergeless Chain for shorter lengths.

            As by definition all horizontal subchains of a horizontal quasi line consist of at least $3$ robots, no merge can be performed on them.
            In order to connect two horizontal or two vertical quasi lines without enabling a merge, they must be connected by so called \emph{stairways}.
            Stairways can have arbitrary length and are subchains of alternating left and right turns.
            Figure~\ref{fig:stairway3} shows an example.
            \begin{figure}[h]
                \centering
                \includegraphics[scale=\figscale]{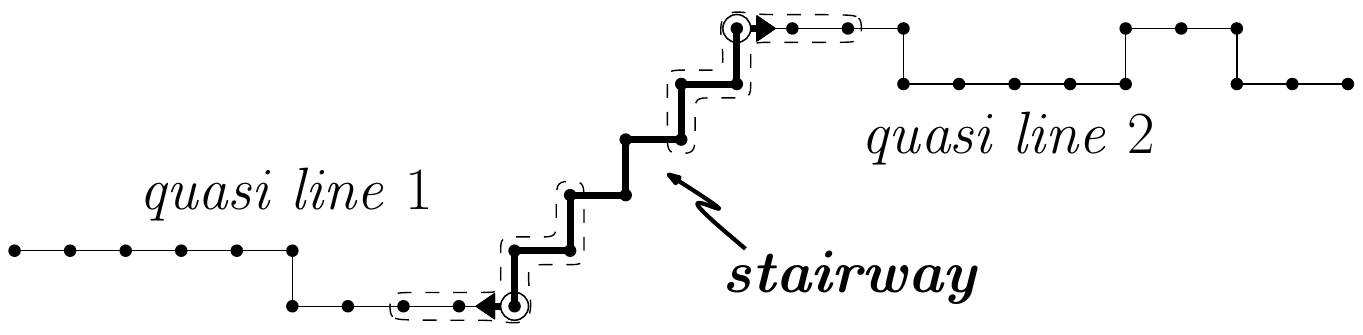}
                \caption{Two quasi lines, connected by a \emph{stairway}.}
                \label{fig:stairway3}
            \end{figure}
            All differently shaped connecting subchains would allow merges.
            If a horizontal and a vertical quasi line are connected, then this can also be done without any stairway.

            The above construction exactly leads to the run starting subchains $(i,ii)$ of Figure~\ref{fig:runstartingrobots}.
            In Figure~\ref{fig:stairway3}, such subchains are bordered by dashed curves.
            So new runs are always started at the endpoints of quasi lines.

            For being able to close the chain, there must exist both, horizontal and vertical quasi lines.
            We start on the left of Figure~\ref{fig:goodpair_existence} at a robot $s$ where a vertical and a horizontal quasi line
            are neighbors.
            For simplicity, we assume that the stairways, connecting the quasi lines in the figure, are of minimum length and symbolize
            quasi lines by dashed line segments.
            The fat robots correspond to the fat robots of Figure~\ref{fig:goodpair_ex1}, i.e., are the outer chain neighbors of quasi lines.
            As in Figure~\ref{fig:goodpair_ex1}, the run pairs are good pairs if the fat robots are both located on the same side.
            So, if no good pair existed, they must instead lie on alternating sides of the quasi lines.
            \begin{figure}[h]
                \centering
                \includegraphics[scale=\figscale]{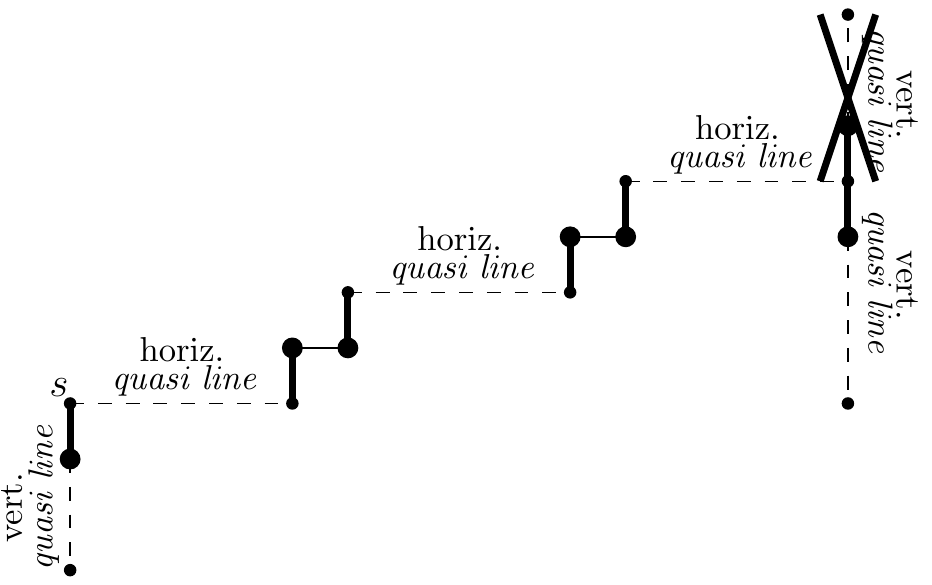}
                \caption{Proof idea: It is impossible that always both vertical quasi lines point to opposite directions. Then, in the example at the last horizontal quasi line, a good pair starts.}
                \label{fig:goodpair_existence}
            \end{figure}
            Figure~\ref{fig:goodpair_existence} shows how a sequence of horizontal quasi lines then must look like.
            Because we want to close the chain, we need a second vertical quasi line.
            If no merge is possible, this quasi line must point upwards, i.e., in opposite direction than the first one.
            In this proof, we show by contradiction that if this always is the case, then the chain cannot be closed.
            \begin{figure}[h]
                \centering
                \includegraphics[scale=\figscale]{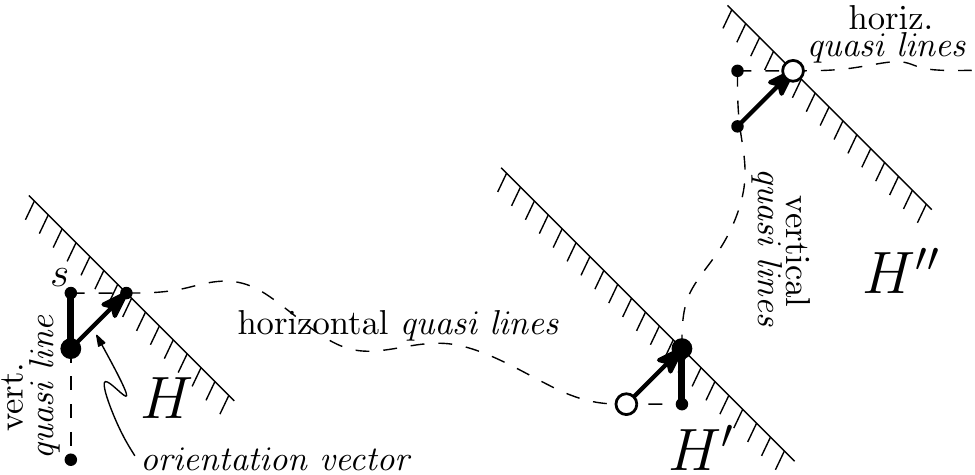}
                \caption{In a Mergeless Chain, good pairs do always exist.}
                \label{fig:goodpair_existence2}
            \end{figure}
            
            In Figure~\ref{fig:goodpair_existence2}, we have grouped subchains of connected horizontal respectively vertical quasi lines and symbolize such a group by a dashed curve.
            For our argumentation, we take the robot $s$ which is the first robot of the first horizontal quasi line, as the starting point.
            We define an \emph{orientation vector} which points from the second last robot
            of the last vertical quasi line to the second robot of the subsequent horizontal one.
            (If between these two quasi lines a stairway exists, then the orientation vector becomes longer than in the figure, but still points to the same direction.)
            
            The orientation vector defines the half-plane $H$ by orthogonally pointing from the interior to its boundary.
            In order to close the chain, we must return to $s$ and for this, reenter $H$.
            This cannot be done, using only horizontal quasi lines.
            So, we look at the point where the first vertical quasi line of the next group of vertical quasi lines occurs.
            Here, the orientation vector and the half-plane $H'$ are defined analogue to the previous ones.
            Because we assumed that the fat black robots lie on different sides of the horizontal quasi lines, the vertical quasi lines
            must point upwards.
            Then, the orientation vectors are parallel and $H\subset H'$.
            If we continue the argumentation for the vertical quasi lines, we have to ensure that now the fat white robots are located at different sides.
            Then the following horizontal quasi lines must point to the right again.
            Then again the orientation vectors are in parallel and we get $H\subset H'\subset H''$.
            Continuing this construction, we never get back to the robot $s$ and so can never close the chain.
            So, a good pair must exist. Because we have assumed that the chain has been a Mergeless Chain also during the previous $L-1$ rounds, this good pair is a progress pair.
        \end{proof}
    \subsection{Proof of Lemma~\ref{lem:progresspair_merge}}\label{ssec:progresspair_merge}

        For the proof, we need the run invariants of Lemma~\ref{lem:runinvariants}.
        \begin{lemma}\label{lem:runinvariants}
            The value $L=\pipelininginterval$ and the value $\viewingradius$ for the viewing path length ensure, that for every run $S$ until it terminates, the following invariants holds.
            \begin{enumerate}
                \item Every round, $S$ moves one robot further in moving direction.\label{enum:runinvmovement}
                \item After the first three rounds after its start $S$ is always located on a quasi line. (I.e., the reshapements of its runner do not violate the quasi line definition~\ref{def:quasiline}.)\label{enum:runinvquasiline}
                \item $S$ cannot see other sequent runs in front of it.\label{enum:runinvvisseqrun}
                \item $S$ is either in progress of passing along another run or the runner $R(S)$ executes one of the operations of Figure~\ref{fig:ALG_hop_simple2}.\label{enum:runinvops}
                \item Good pairs stay being good pairs.\label{enum:runinvgoodpairs}
            \end{enumerate}
        \end{lemma}
        Now, we can prove Lemma~\ref{lem:progresspair_merge}:
        \begin{proof}%(Lemma~\ref{lem:progresspair_merge})
            $\ref{enum:progpairmerge}):$
            At the time, a progress pair $S,S'$ is started, the subchain connecting both is a quasi line.
            Because of Lemma~\ref{lem:runinvariants}.$\ref{enum:runinvquasiline})$, this also does not change
            if other runs are located on this subchain.
            And also merges preserve the quasi line properties.
            So, if not stopped, $S$ and $S'$ keep moving towards each other (Lemma~\ref{lem:runinvariants}.$\ref{enum:runinvmovement})$).
            Because of Lemma~\ref{lem:runinvariants}.$\ref{enum:runinvgoodpairs})$ a merge can be performed at the latest when they meet, which takes at most $n$ rounds.

            It remains to show that $S,S'$ are not stopped by the algorithm's termination conditions (Table~\ref{table:runterminate}) before the merge could be performed.
            In the following, we check all these conditions:
            $\ref{enum:stoptooclose})$: By definition of a progress pair, no merge has been performed since the last time new runs have been started.
            This means, that the distance between $S$ and its next sequent runs in front of it is at least $L-1$ (cf.\ the proof of Lemma~\ref{lem:runinvariants}.$\ref{enum:runinvops})$) which is bigger than the viewing path length. (The same holds for $S'$.)
            So the runs of a progress pair cannot be stopped by this termination condition.

            $\ref{enum:stopendpoint})$: If a run of a progress pair can see an endpoint of the quasi line in front of it, then the other run of the progress pair must have previously been stopped.
            Because a run of a progress pair cannot be stopped by condition $\ref{enum:stoptooclose})$, it instead must have been stopped because of a merge.
            So, a termination because of condition $\ref{enum:stopendpoint})$ is allowed.

            $\ref{enum:stoppass},\ref{enum:stoplongop})$: If a run $S$ of a progress pair is stopped because of this, then its target corner
            was removed.
            This could either have happened because of the reshapement of a sequent run or because of a merge operation.
            Because we have ensured the minimum distance of two sequent runs to be large enough, the first case cannot happened.
            So, a merge must have been the reason.
            Then, this merge must have been enabled by another run $S^\star$, moving towards $S$. 
            If $S^\star$ is the partner run of $S$, then stopping is allowed.
            Else, because progress pairs are nested into each other, $S^\star$ either was not a run of a progress pair or its partner run has previously been stopped.
            In the latter case, because runs of progress pairs cannot be stopped because of condition $\ref{enum:stoptooclose})$,
            the run must have been stopped by an earlier merge.
            In both cases, the current merge can be credited to the progress pair of $S$.

            $\ref{enum:stopmerge})$: By the same arguments, this merge does not need to be also credited to a different progress pair.
            This then also proves $\ref{enum:progpairuniqueness})$ of the lemma.
        \end{proof}
        Now, we provide the omitted proof of Lemma~\ref{lem:runinvariants}.
        \begin{proof}
            $\ref{enum:runinvmovement})$: This directly follows by the run definition in Subsection~\ref{ssec:runs}.

            $\ref{enum:runinvquasiline})$: Cf.\ Figure~\ref{fig:ALG_hop_simple2}.
            $a)$ ensures that on horizontal quasi lines only horizontal subchains of lengths $>2$ are shortened and vertical subchains remain unextended.
            $b)$ does not reshape.
            $c)$ can only executed during the first three rounds after $S$ has been started.

            $\ref{enum:runinvvisseqrun})$: This is ensured by the run termination condition of Table~\ref{table:runterminate}.$\ref{enum:stoptooclose})$.

            $\ref{enum:runinvops})$: Cf.\ Figure~\ref{fig:ALG_hop_simple2}.
            All these operations ensure that if $S$ was located at some corner $c_1$ when the operation started,
            it afterwards either is located at some other corner $c_2$ such that both corners are rotated equally or terminates if the target
            corner has been removed during the operation (cf.~Table~\ref{table:runterminate}.$\ref{enum:stoplongop},\ref{enum:stoppass}))$.
            Then, because still located on a quasi line (cf.~$\ref{enum:runinvquasiline})$), again $a)$ or $b)$ can be applied or run passing is started.
            
            The run passing needs some closer look, because it interrupts other operations.
            In order to ensure a regulated behavior, we chose the distance between sequent runs big enough 
            such that a run does not have to execute a new run passing operation before it has finished its previous one.
            We settle the values for the constant $L$ and the viewing path length appropriately.
            We look at two sequent runs $S_1$ and $S_1^{\mathrm{succ}}$ such that $S_1^{\mathrm{succ}}$ has been started after $S_1$.
            Their distance $D$ is at least $L-1$.
            This value is achieved if at the start of $S_1$ the operation $c)$ (Figure~\ref{fig:ALG_hop_simple2}) was executed.
            
            Now, we chose the value for the constant $D$ big enough for ensuring that while some run $S_2$ is passing along $S_1$,
            $S_1^{\mathrm{succ}}$ becomes visible to $S_2$ the earliest when the passing operation with $S_1$ has been completed.
            Figure~\ref{fig:crossconst} shows an example for the longest possible duration of a run passing operation.
            Here, it takes $6$ rounds until $S_2$ has arrived at its target corner.
            Because the run passing operation starts when the distance between $S_1$ and $S_2$ is $\leq\crossdistance$, after the
            passing operation, the distance between $S_2$ and $S_1^{\mathrm{succ}}$ equals $D-9$.
            We want this then still be $\geq\crossdistance$.
            So we choose $D\geq 12$ and together with the above argumentation concerning the minimum distance between sequent runs
            follows $L\geq 13$.
            In order to detect that the distance has become smaller than $12$ (and solve this problem), the viewing path length must be $11$ (cf.\ Table~\ref{table:runterminate}.$\ref{enum:stoptooclose})$).

            $\ref{enum:runinvgoodpairs})$: When defining good pairs in Subsection~\ref{ssec:runs},
            good pairs have been characterized by the relative position of the outer direct neighbors of the good pair according to the quasi line.
            The first part of the proof of $\ref{enum:runinvops})$ finishes the proof.
        \end{proof}
\bibliography{references}
\end{document}